\documentclass[preprint,12pt]{elsarticle}

\usepackage{amssymb,amsmath}
\usepackage{amsthm}
\newtheorem{theorem}{Theorem}

\newtheorem{lemma}[theorem]{Lemma}
\newtheorem*{problem}{Problem}
\newtheorem{corollary}[theorem]{Corollary}
\newtheorem{proposition}[theorem]{Proposition}
\usepackage{algorithm,caption}
\usepackage{algpseudocode}
\usepackage{todonotes}
\usepackage{enumitem}
\usepackage{polski}
\usepackage{latexsym}
\usepackage{amsfonts}

\def\ll{,\ldots,}

\journal{Discrete Applied Mathematics}
\usepackage[english]{babel}

\begin{document}

\begin{frontmatter}



\title{$H$-colouring $P_t$-free graphs in subexponential time}

\author[oxf]{Carla Groenland} 
\ead{groenland@maths.ox.ac.uk}
\author[war]{Karolina Okrasa}
\ead{okrasak@student.mini.pw.edu.pl}
\author[war]{Pawe\l ~Rz\k{a}\.{z}ewski}
\ead{p.rzazewski@mini.pw.edu.pl}
\author[oxf]{Alex Scott\fnref{fn1}}
\ead{scott@maths.ox.ac.uk}
\author[pri]{Paul Seymour\fnref{fn2}}
\ead{pds@math.princeton.edu}
\author[pri]{Sophie Spirkl}
\ead{sspirkl@math.princeton.edu}

\fntext[fn1]{Supported by a Leverhulme Trust Research Fellowship.}
\fntext[fn2]{Supported by ONR grant N00014-14-1-0084 and NSF grant DMS-1265563.}
\address[oxf]{Mathematical Institute, University of Oxford, Oxford, OX2 6GG, United Kingdom.}
\address[war]{Faculty of Mathematics and Information Science, Warsaw University of Technology, Warsaw, 00-661, Poland.}
\address[pri]{Princeton University, Princeton, NJ 08544, USA}

%
%

\begin{abstract}
A graph is called $P_t$-free if it does not contain the path on $t$ vertices as an induced subgraph. 
Let $H$ be a multigraph with the property that any two distinct vertices share at most one common neighbour. We show that the generating function for (list) graph homomorphisms from $G$ to $H$ can be calculated in subexponential time  $2^{O\left(\sqrt{tn\log(n)}\right)}$ for $n=|V(G)|$ in the class of $P_t$-free graphs $G$. As a corollary, we show that the number of 3-colourings of a $P_t$-free graph $G$ can be found in subexponential time. On the other hand, no subexponential time algorithm exists for  4-colourability of $P_t$-free graphs assuming the Exponential Time Hypothesis. Along the way, we prove that $P_t$-free graphs have pathwidth that is linear in their maximum degree.
\end{abstract}


\end{frontmatter}



\section{Introduction}
Throughout this paper, graphs do not have multiple edges or loops. When we need general graphs (with multiple edges and loops), we call them multigraphs. We use the notation $vv'$ for the edge $\{v,v'\}$. For a multigraph $G$, the set $N_G(v)=\{v': vv'\in E(G)\}$ contains $v$ if and only if $G$ has a loop at vertex $v$. 

A \emph{$k$-colouring} of a graph $G$ is a function $c:V(G) \to \{1,\dots, k\}$ such that $c(v) \neq c(v')$ for all $vv' \in E(G)$. 
The decision problem \textsc{$k$-colourability} asks whether a given graph $G$ is $k$-colourable. This problem is NP-complete for $k\geq 3$ in general. 
In order to investigate what graph structure makes the decision problem hard, a natural question is whether the problem becomes easy if it is restricted to instances that do not contain a particular structure.
Thus we restrict to the class of $F$-free graphs, that is, those graphs which do not contain $F$ as an induced subgraph, for some fixed graph $F$. If $k$ is part of the input, a full classification is given by Kr\'{a}l et al. \cite{Kral01ea}. For fixed $k\geq 3$, \textsc{$k$-colourability} is shown to be NP-complete for $F$-free graphs if $F$ is either a cycle $C_\ell$ for $\ell\geq 3$ \cite{Kaminski07Lozin} or the claw $K_{1,3}$ \cite{Holyer81,Leven83Galil}. Any graph $F$ which does not contain a cycle nor the claw is a disjoint union of paths. 

Let $P_t$ denote the path on $t$ vertices. Polynomial-time algorithms for deciding \textsc{$k$-colourability} for $P_t$-free graphs exist for $t\leq 5$ \cite{Hoang10ea}, $(k,t)=(4,6)$ \cite{Chudnovsky18ea,Chudnovsky18eaii} and $(k,t)=(3,7)$ \cite{Bonomo17ea}. On the other hand, Huang \cite{Huang16} showed \textsc{$4$-colourability} is NP-complete for $P_7$-free graphs and \textsc{$5$-colourability} is NP-complete for $P_6$-free graphs. It is an open problem to determine the complexity of \textsc{3-colourability} of $P_t$-free graphs for $t\geq8$. 

It is also open whether \textsc{maximum independent set} is decidable in polynomial time for $P_t$-free graphs for $t\geq 7$. Brause \cite{Brause17} and Basc\'{o}, Marx and Tuza \cite{Basco17ea} independently showed a greedy exhaustive approach yields a subexponential-time algorithm for \textsc{maximum independent set} on $P_t$-free graphs. In this paper we show that there are subexponential-time algorithms for a larger class of problems, including \textsc{$3$-colourability} and \textsc{maximum independent set}, and also give counting results. Our algorithm builds on the following property of $P_t$-free graphs.
\begin{lemma}
\label{lem:pathdecomp}
A $P_t$-free graph of maximum degree $\Delta$ has pathwidth at most $(\Delta-1)(t-2)+1$. Moreover, a path-decomposition of this width can be found in polynomial time.
\end{lemma}
This lemma is an improvement on the treewidth bound of Basc\'{o} et al. \cite{Basco18ea} for $P_t$-free graphs (which they used to improve the algorithm of Basc\'{o} et al. \cite{Basco17ea}). The result is tight up to a constant factor, which can be seen by replacing the vertices of a binary tree by cliques.

We now introduce our framework. Let $H$ be a multigraph and $G$ a (simple) graph. A \emph{graph homomorphism} from $G$ to $H$ is a map $f:V(G)\to V(H)$ such that $vv' \in E(G)$ implies $f(v)f(v')\in E(H)$. (Thus a 3-colouring is a graph homomorphism to $K_3$.)  
A \emph{list $H$-colouring instance} $I=(G,L)$ consists of a graph $G$ together with a function $L:V(G)\to \mathcal{P}(V(H))$ that assigns a subset $L_{v} \subseteq V(H)$ to every $v\in V(G)$. A \emph{list $H$-colouring} of such an instance is a graph homomorphism $f$ from $G$ to $H$ such that $f(v)\in L_{v}$ for all $v\in V(G)$. We denote the set of list $H$-colourings of $(G,L)$ by $\mathcal{LC}((G,L),H)$.

A useful way to summarise information about $H$-colourings of a graph $G$ is to use a multivariate generating function (see for example \cite{Scott09Sorkin}). 
Given a multigraph $H$, we define the  \emph{partition function} $p_{(G,L)\to H}(x)$ by
\[
p_{(G,L)\to H}(x):=
\sum_{f\in \mathcal{LC}((G,L),H)}  
\prod_{u\in V(G)}w_{u,f(u)} x_{f(u)}
.
\]
We omit the lists $L$ where clear from context.
The weights $w_{v,h}$ (for $v \in V(G)$ and $h\in V(H)$) are included to allow more general application and can be ignored by choosing them identically one. For $H=(\{h,h'\},\{hh',h' h'\})$, $p_{G\to H}(x)$ gives the independent set polynomial when $x_{h'}$ is set to one. 

Summing appropriate coefficients of the polynomial, the partition function can be used to for example count the number of list $H$-colourings,
or to count the number of ``restrictive $H$-colourings'' \cite{Diaz05ea} in which a restriction is placed on the size of the preimages of the vertices of $H$.

The following theorem is our main result and will be proved in Section \ref{sec:alg}.
\begin{theorem}
\label{thm:main}
Let $H$ be a multigraph so that $|N_H(h)\cap N_H(h')|\leq 1$ for all distinct vertices $h,h'$ of $H$. For $t\geq 4$, the polynomial $p_{G\to H}(x)$ can be calculated for every $P_t$-free graph $G$ in time $2^{O\left(\sqrt{tn\log(n)}\right)}$ where $n=|V(G)|$. 
\end{theorem}
For simple graphs $H$, the condition $|N_H(h)\cap N_H(h')|\leq 1$ for all distinct $h,h'$ is equivalent to $H$ not having $C_4$ as (not necessarily induced) subgraph.
\begin{corollary}
\label{cor}
The following problems can be solved for $P_t$-free graphs $G$ in time $2^{O\left(\sqrt{tn\log(n)}\right)}$ where $n=|V(G)|$:
\begin{itemize}
\item Counting the number of $H$-colourings for any fixed simple graph $H$ with no $C_4$ subgraph.
\item Computing the independent set polynomial. 
\end{itemize}
In particular, it can be decided in subexponential time whether a $P_t$-free graph is 3-colourable (and the number of 3-colourings can be counted). 
\end{corollary}
For example, if $G$ is $P_t$-free and $H$ is an odd cycle, then we can count the number of $H$-colourings of $G$ in subexponential time. This problem is $\#P$-complete if all graphs $G$ are allowed \cite{Dyer00Greenhill} and the corresponding decision problem is NP-complete \cite{Hell90Nesertil}.

The Exponential Time Hypothesis (ETH) \cite{Impagliazzo01Paturi} states that there is an $\epsilon>0$ such that there is no $O(2^{\epsilon n})$-time algorithm for \textsc{3-sat}. In Section \ref{sec:ext} we prove the following result that shows that it is in some sense unlikely that Corollary \ref{cor} will extend to $k$-colouring $P_t$-free graphs for $k>3$.
\begin{proposition}
\label{prop:ETH}
If the Exponential Time Hypothesis is true, then there is no algorithm running in time subexponential in the number of vertices of the graph for
\begin{itemize}
    \item \textsc{$4$-colourability} on $P_7$-free graphs;
    \item \textsc{$3$-colourability} for $F$-free graphs for any connected $F$ which is not a path.
\end{itemize}
\end{proposition}
This result might shed some light on why the complexity status of \textsc{$3$-colourability} of $P_t$-free graphs for $t$ large has remained open whereas the complexity of \textsc{$k$-colourability} of $F$-free graphs has been settled for other values of $k$ or $F$.



\section{Pathwidth of $P_t$-free graphs and dynamic programming}
\label{sec:pathdecompo}
A {\em path-decomposition} of a graph $G$ is a sequence of subsets $X_i$ of vertices of $G$ with three properties:
\begin{itemize}
\item The vertex set of $G$ equals $\bigcup_i X_i$,
\item For each edge of $G$, there exists an $i$ such that both endpoints of the edge belong to subset $X_i$, and
\item $X_\ell\cap X_j\subseteq X_i$ for every three indices $1\le \ell \le i\le j$.
\end{itemize}
The {\em pathwidth} of $G$ is defined as the minimum of $\max_i|X_i|-1$ over the path-decompositions of $G$.

Let $T$ be a rooted tree with root $r$.
For a vertex $v$,  let $T_v$ denote the path of $T$ between $r$ and $v$.
For each vertex $w$, fix a linear order of the set of children of $w$. We call this a {\em plane tree}. 
If $u,v$ are both children of $w$ and $u$ is earlier than $v$
in the corresponding ordering, we say $u$ is an {\em elder sibling} of $v$.

Let $T$ be a spanning tree of a graph $G$, equipped with orders to make it a plane tree. We call it an {\em uncle tree}
of $G$ if for every edge $uv$ of $G$ that is not an edge of $T$, one of $u,v$ has an elder sibling that is an ancestor of the other.
We use the following result of Seymour \cite{Seymour16}.
\begin{theorem}
\label{thm:seymour}
For every connected graph we can compute an uncle tree with any specified vertex as root in polynomial time.
\end{theorem}
(The proof is easy; grow a depth-first tree, subject to the condition that the path of the tree between each vertex and 
the root is induced.)

\begin{proof}[Proof of Lemma \ref{lem:pathdecomp}]
We may assume that $G$ is connected.
Take an uncle tree $T$ of $G$, and order its leaves as $p_1\ll p_k$ say, in the natural order of the leaves of a plane tree. 
For $1\le i\le k $, let $X_i$ be the set of vertices
of $T$ that either belong to $T_{p_i}$, or have an elder sibling (and hence also a parent) in this set.
We claim that the sequence $(X_1\ll X_k)$ is a path-decomposition of $G$.
We check that:
\begin{itemize}
\item Every vertex belongs to some $X_i$ (this is clear). 
\item For all $u,v$ adjacent in $G$, there exists $i$ with $u,v\in X_i$. To see this, since $T$ is an uncle tree we may assume that
$u$ has an elder sibling $u'$ that is an ancestor of $v$. Choose a leaf $p_i$ of $T$ such that $T_{p_i}$ contains $v$;
then the common parent of $u,u'$ belongs to $T_{p_i}$, and hence $u,v\in X_i$.
\item  $X_\ell \cap X_j\subseteq X_i$ for $1\le \ell \le i\le j\le k$. To see this, let $v\in X_\ell \cap X_j$;
consequently either $v$ or an elder sibling of $v$ belongs to $T_{p_\ell}$,
and either $v$ or an elder sibling belongs to $T_{p_j}$. It follows that either $v$ or an elder sibling of $v$ belongs to $T_{p_i}$,
from the ordering of the leaves of $T$.
\end{itemize}
This proves the claim.

If $T$ is an uncle tree of $G$, then each path of $T$ starting from $r$ is induced in $G$; and so if $G$ does not contain
$P_t$, then $|T_{p_i}|\leq t-1$ for each $i$, and so $|X_i|\leq \Delta-1+(t-3)(\Delta-2)+1$, where $\Delta$ denotes the maximum degree of $G$.
\end{proof}
We give a short outline of how the standard dynamic programming approach can be applied to compute $p_{G\to H}(x)$ in time $2^{O(p)}n$ given a path-decomposition of width $p$ of a graph $G$ on $n$ vertices.

Let $(X_1,\dots,X_k)$ be the given path-decomposition of $G$. We define $X(i)=\bigcup_{j\leq i}X_j$. For each $i\in [k]$ and list $H$-colouring $g:X_i\to V(H)$, we compute  the polynomial $p_{G[X(i)]\to H}$ with the vertices in $X_i$ precoloured: we define $p_i(g)$ as the sum, over the list $H$-colourings $f:X(i)\to V(H)$ such that $f|_{X_i}=g$, of \[
\prod_{u\in X(i)}w_{u,f(u)} x_{f(u)}
.
\]
For each list $H$-colouring $g:X_1\to V(H)$, we set $p_1(g)=\prod_{v \in X_1} w_{v,g(v)}x_{g(v)}$. Having computed all $p_i(g)$ for some $i$, for each list $H$-colouring $g:X_{i+1}\to V(H)$ we select the list $H$-colourings $g_1,\dots,g_\ell$ of $X_i$ that are compatible with $g$, that is, $g_i(v)=g(v)$ for $v\in X_i\cap X_{i+1}$ and $g(u)g(v)\in E(H)$ if $uv\in E(G)$ for all $u\in X_i$ and $v\in X_{i+1}$. Since $N_G[v]\cap X(i)\subseteq X_i$ for all $v\in X_{i+1}\setminus X_i$, we can then compute
\[
p_{i+1}(g)= \sum_{j=1}^\ell p_i(g_j) \prod_{v\in X_{i+1}\setminus X_i} w_{v,g(v)}x_{g(v)}.
\]
Finally, we calculate the desired $p_{G\to H}$, which is the sum, over all list $H$-colourings $g$ of $X_k$, of $p_k(g)$.

\section{Algorithm and time analysis}
\label{sec:alg}
Throughout this section, $H$ is a fixed multigraph such that $|N_H(h)\cap N_H(h')|\leq 1$ for all distinct $h,h'$ in $H$. We allow loops in $H$ but no multiple edges.\footnote{It is possible to extend the algorithm to compute a version of $p_{G\to H}(x)$ with edge weights $A_{h,h'}$ for $h,h'\in V(H)$, but in this case the weights $w_{v',h'}$ have to be updated in Line 3 and 6 of Algorithm HCol to $w_{v',h'} A_{h,h'}$ for all $v'\in N_{v}$.} The $P_t$-free graphs $G$ are assumed to be simple.


We shall say a list colouring instance $I=(G,L)$ has \emph{weight} $w(I)=\sum_{v\in V(G)}|L_{v}|$ and is \emph{reduced} if $|L_{v}|\geq 2$ for all $v\in V(G)$. 
The key observation we need is the following.
\begin{lemma}
\label{lem:neighbourhoodlists}
Let $I=(G,L)$ be a reduced list $H$-colouring instance and let $v\in V(G)$ with degree $d(v)$. For $h \in V(H)$, let 
\[
C_{h} =\{v' \in N_G(v): L_{v'} \subseteq N_H(h)\}.
\]
Then there is at most one $h \in L_{v}$ for which $|C_{h}|> \frac12 d(v)$. 
\end{lemma}
\begin{proof}
Suppose $h \neq h'$ in $L_{v}$ both satisfy $|C_{h}|,|C_{h'}|>\frac12 d(v)$. Then there exists $v' \in N_G(v)$ such that $v' \in C_{h} \cap C_{h'}$. Hence $L_{v'}\subseteq N_H(h)\cap N_H(h')$, so that by our assumption on $H$ we find $|L_{v'}|\leq 1$, contradicting the assumption that $I$ is reduced.
\end{proof}
This lemma tells us that ``colouring'' a vertex $v$ of degree $d(v)=\Delta$ decreases the weight of a reduced instance by at least $\frac12\Delta$ for all but one ``colour'' in $L_{v}$. Either there is a vertex of high degree and we can reduce the weight significantly by colouring this vertex, or $\Delta$ is ``small'' and we can apply the results from the previous section to compute $p_{G\to H}(x)$ in time $2^{O(t\Delta)}$. 

Our algorithm ``HCol'' for computing the list $H$-colouring function of a graph $G$ is given below. This algorithm either terminates or recurses on instances of strictly smaller weight. Therefore, it always terminates in finite time. We can represent the recursions by a tree: the root is the first call of the algorithm and each recursive call creates a child. For $P_t$-free graphs, we can bound the number of nodes in this recursion tree.
\begin{proposition}
\label{prop:recursivecalls}
Let $t\geq 4$, $c>4\sqrt{t|V(H)|}/\log(2)$ and $f(w)=2^{c\sqrt{w\log(w)}}$. Then there exists an $n_0$ for Algorithm HCol such that if it is applied to an instance  $I=(G,L)$ of weight $w(I)$ with $G$ a $P_t$-free graph, then the number of nodes in the corresponding recursion tree is bounded by $f(w(I))$.   
\end{proposition}
\begin{algorithm}
Input: a list $H$-colouring instance $I=(G,L)$ for $G=(V,E)$.
\begin{enumerate}[nolistsep,leftmargin=0.5cm]
    \item If $|V|\leq n_0$, compute the list $H$-colouring function exhaustively. 
    \item If there exists $v\in V$ such that $|L_{v}|=0$, return 0. 
    \item If there exists $v\in V$ such that $|L_{v}|=1$, say $L_{v}=\{h\}$, then set $L'_{v'} = L_{v'}\cap N_H(h)$ for $v' \in N_G(v)$ and $L'_{v'}=L_{v'}$ for $v' \not \in N_{G}(v)$. 
    Return $w_{v,h}x_{h} $HCol($G-v,L')$. 
    \item If $G$ is not connected, let $G_1,\dots,G_k$ be the connected components. Return $\prod_{i=1}^k$HCol$(G_i,L|_{V(G_i)})$.
    \item If the maximum degree of $G$ is at most $\sqrt{n\log(n)/t}$, compute a path-decomposition of $G$ of width $O\left(\sqrt{tn\log(n)}\right)$ and compute the result using dynamic programming.
    \item Otherwise take $v\in V$ of maximal degree. For $h\in L_{v}$, set
    $L^{h}_{v'}=L_{v}\cap N_H(h)$ if $v'\in N_G(v)$ and $L^h_{v'}=L_{v'}$ if $v' \not\in N_G(v)$. Return $\sum_{h \in L_{v}}w_{v,h}x_h$ HCol$(G-v,L^{h})$. 
\end{enumerate}
\caption*{\textbf{Algorithm HCol}: Outputs the list $H$-colouring function.}
\end{algorithm}
Algorithm HCol gives the correct answer for any graph $G$: in line 2 we note that if some vertex has an empty list, then $p_{G\to H}=0$; in line 3 we note that if the list of a vertex $v \in V(G)$ has a single element $h \in V(H)$, then $v$ has to be mapped to $h$, i.e. $p_{G\to H}(x)=w_{v,h}x_h p_{G-v\to H}(x)$; in line 4 we use the algebraic identity $p_{G_1\sqcup G_2 \to H}(x)=p_{G_1 \to H}(x)p_{G_2\to H}(x)$; in line 6, we use $p_{G \to H}(x)=\sum_{h \in L_{v}} w_{v,h}x_{h} p_{G-{v}\to H}(x)$.\footnote{We left the lists of the vertices implicit in the notation; the precise way in which the lists need to be updated is given in the algorithm.}

Inspecting and updating the lists of vertices, finding a vertex of maximal degree and finding the connected components of a graph on $n$ vertices can all be done in time $Cn^2$, where the constant $C$ may depend on $H$ and $n_0$. Line 5 is applied at most once {per node} and takes $2^{O\left(\sqrt{tn\log(n)}\right)}$. Since $|L_v|\leq |V(H)|$ for all $v\in V(G)$, it follows that $w(I)\leq |V(H)|n=O(n)$. Theorem \ref{thm:main} follows from Proposition \ref{prop:recursivecalls} by observing that $Cn^22^{O\left(\sqrt{tn\log(n)}\right)}2^{O\left(\sqrt{tn\log(n)}\right)}=2^{O\left(\sqrt{tn\log(n)}\right)}$.

We require the following simple estimate.
\begin{lemma}
\label{lem:calculus}
Let $m,y>0$ and $c>y^{-1}$. There exists an $n_0\in \mathbb{N}$ such that $f(w)=e^{c\sqrt{w\log(w)}}$ satisfies
\[
f(n-2) + mf\left(n-y\sqrt{n\log(n)}\right)\leq f(n)
\]
for all $n\geq n_0$.
\end{lemma}
\begin{proof}
Let $\epsilon>0$ be given such that $\sqrt{1-x}\leq 1-\frac12x$ and $e^{-x}\leq 1-x/2$ for all $0\leq x\leq \epsilon$. Choose $n_0$ sufficiently large such that $2/n,y\sqrt{\log(n)/n},c\log(n)/\sqrt{n}\leq\epsilon$ and $mn^{-cy/2}< \frac{c}2\sqrt{\log(n)/n}$ for all $n\geq n_0$ (by assumption, $cy>1$). We calculate
\begin{align*}
f(n-2)+mf\left(n-y\sqrt{n\log(n)}\right) 
&\leq e^{c\sqrt{(n-2)\log(n)}}+me^{c\sqrt{\log(n)}\left(n-y\sqrt{n\log(n)}\right)^{\frac12}}\\
&= f(n)^{\sqrt{1-2/n}}+mf(n)^{\left(1-y\sqrt{\log(n)/n}\right)^{\frac12}} \\
&\leq f(n)^{1-1/n}+mf(n)^{1-\frac12y\sqrt{\log(n)/n}} \\
&=f(n)\left[e^{-c\sqrt{\log(n)/n}}+me^{-\frac{1}2cy\log(n)}\right]
\end{align*}
But
\[
e^{-c\sqrt{\log(n)/n}}+me^{-\frac{1}2cy\log(n)} \leq 1-\frac{1}2c\sqrt{\log(n)/n}+mn^{-\frac12cy}< 1. \qedhere
\]
\end{proof}

\begin{proof}[Proof of Proposition \ref{prop:recursivecalls}]
Let $n_0$ be given from Lemma \ref{lem:calculus} applied with $m=|V(H)|$, $y=\frac1{4\sqrt{t|V(H)|}}$ and using $c \log(2)$ instead of $c$. Enlarging $n_0$ if necessary for the last three properties, we may now assume that 
\begin{align*}
f(w-2) + |V(H)| f\left(w-y\sqrt{w\log(w)}\right) &\leq f(w) & \text{ for all }&w\geq n_0,\\
f(k)+f(\ell)+1 &\leq f(k+\ell) & \text{ for all }&k,\ell\geq n_0, 
\\
f(k)+ (w-k)+1&\leq f(w) & \text{ for all }&w\geq n_0,~k< w,\\
w+1 &\leq f(w)& \text{ for all }&w\geq n_0.
\end{align*}

The proposition is proved by induction on $w=w(I)$. If the algorithm terminates in line 1, 2 or 5, then there is only one iteration and $f(n)\geq 1$ for all $n\in \mathbb{Z}_{\geq 0}$. If the algorithm reaches line 3, then $G$ has at least $n_0$ vertices and $|L_{v}|\geq 1$ for all $v\in V$, so that $w(I)\geq n_0$. Therefore, the statement holds for all $w<n_0$.

Suppose the statement has been shown for instances with $w(I)<w$ for some $w\geq n_0$. If the algorithm recurses on line 3, then the removed vertex contributed at least 1 to the weight, and so by induction at most $f(w-1)+1\leq f(w)$ iterations are taken. 

If the algorithm reaches line 4, we may assume the instance is reduced, and so $w(I)\leq 2|V(G)|$.
Suppose the graph $G$ is disconnected with connected components $G_1,\dots, G_k$. Let $I_i=(G_i,L|_{V(G_i)})$ and note that $I_i$ is also reduced. Hence if $|V(G_i)|\leq \frac12 w(I_i)\leq  n_0$, then the recursive call on $G_i$ will take a single iteration. Renumber so that $I_1,\dots,I_\ell$ have weight at most $2n_0$ and $I_{\ell+1},\dots,I_k$ have weight at least $2n_0$. By induction the algorithm takes at most 
\[
\sum_{i=1}^\ell 1+ \sum_{i=\ell+1}^k f(w(I_i))+1\leq f(w)
\]
iterations, where the inequality follows from the assumptions we placed on $n_0$, considering $k-\ell=0$, $k-\ell=1$ and $k-\ell>1$ separately.

At line 6, the hypothesis of Lemma \ref{lem:neighbourhoodlists} is satisfied. All but one of the instances have their weight reduced by at least \[
\frac1{2}\sqrt{n\log(n)/t}\geq \frac1{2\sqrt{t|V(H)|}} \sqrt{w\left(\log(w)-\log(|V(H)|)\right)}\geq y\sqrt{w\log(w)}
\]
(using that $w=\sum_{v\in V(G)}|L_{v}|\leq |V(H)|n$ and assuming $\log(w)-\log(|V(H)|)\geq \frac14\log(w)$ by enlarging $n_0$ if necessary).
The other instance has its weight reduced by at least 2, since the vertex $v$ with $|L_{v}|\geq 2$ is removed. By induction, the number of nodes in the recursion tree is at most $f(w-2)+(|V(H)|-1)f\left(w-y \sqrt{w\log(w)}\right)+1\leq f(w)$.
\end{proof}

\section{Extensions}
\label{sec:ext}
Our algorithm can easily be adapted to find the $H$-colouring $f$ of $G$ which minimises the cost $\sum_{v \in V(G)}w_{v,f(v)}$; in particular, \textsc{Minimum Cost Homomorphism} \cite{Gutin08ea} can be solved in subexponential time for $G$ and $H$ as above.  

Note that Lemma \ref{lem:neighbourhoodlists} is the bottleneck for extending the time complexity to, for example, $4$-colouring ($H=K_4$): it is possible that the neighbourhood of a vertex $v$ with $L_{v}=\{1,2,3,4\}$ has mostly neighbours with list $\{1,2\}$, so that both $C_3$ and $C_4$ are large. Assuming the Exponential Time Hypothesis (ETH), this is to be expected in view of Proposition \ref{prop:ETH}: under ETH, no such subexponential time algorithm can exist for \textsc{$4$-colourability} on $P_7$-free graphs. 

\begin{proof}[Proof of Proposition \ref{prop:ETH}]
Huang \cite{Huang16} gives a reduction of an instance of \textsc{3-sat} with $m$ formulas and $n$ variables into an instance of \textsc{4-colourability} for a $P_7$-free graph on $O(n+m)$ vertices. Therefore, any algorithm for \textsc{4-colourability} on $P_7$-free graphs yields an algorithm for \textsc{3-sat} with the same time dependence on the input size.

Let $F$ be a connected graph which is not a path. Then $F$ contains either the claw $K_{1,3}$ or a cycle.
We prove ETH implies that there is no subexponential time algorithm for \textsc{$3$-colourability} of $F$-free graphs. 
The standard reduction (for example \cite[Prop. 2.26]{Goldreich08}) from \textsc{3-sat} to \textsc{3-colourability} creates a graph on $O(n+m)$ vertices. 
Kami\'{n}ski and Lozin \cite{Kaminski07Lozin} reduce \textsc{3-colourability} to \textsc{3-colourability} on graphs of girth at least $g$ (for every $g\geq 3$) by (in the worst case) replacing each vertex by a constant-sized gadget. This handles the case when $F$ contains a cycle. 

For $F$ containing the claw $K_{1,3}$, note that claw-free graphs are a superset of line graphs. Holyer \cite{Holyer81} reduces \textsc{3-sat} to \textsc{3-edge colourability} on $3$-regular graphs. Given $n$ variables and $m$ clauses, constant-sized gadgets are created for each variable and clause; additional components are added to the variable gadgets for each time it occurs in a clause. Since there are at most $3m$ such occurrences, this creates a graph on $O(n+m)$ vertices. Since the graph is 3-regular, the number of edges is $O(n+m)$ as well. Hence the line graph of such a graph will have $O(n+m)$ vertices. Since $3$-colourings of the vertices of the line graph are in one-to-one correspondence with $3$-colourings of the edges of the original graph, we can hence reduce  \textsc{3-sat} to \textsc{3-colourability} of line graphs on $O(n+m)$ vertices.
\end{proof}
If ETH holds, then any polynomial-time reduction from \textsc{3-sat} to a problem with a subexponential algorithm must ``blow up'' the instance size. Our result therefore suggests that, if one attempts to prove NP-completeness of \textsc{3-colourability} for $P_t$-free graphs (for $t$ large) by designing gadgets, it may be necessary either to start from a problem whose instance size has already been ``blown up'' or to use gadgets which are not of bounded size.

Our algorithm only uses the property of $P_t$-free graphs that every induced subgraph has pathwidth $O(t\Delta)$. Seymour \cite{Seymour16} proves that a tree-decomposition of width $O(\ell \Delta)$ can be computed efficiently for graphs that do not contain cycles of length at least $\ell$ as induced subgraph; therefore, our algorithm extends to this class of graphs (after adjusting the standard dynamic programming approach of for example \cite{bodlaender13ea} to our setting in a similar fashion as done in Section \ref{sec:pathdecompo}). This motivates the following question.
\begin{problem}
For fixed $t$, are \textsc{3-colourability} and \textsc{maximum independent set} solvable in polynomial time on graphs that have no induced cycles of length greater than $t$?
\end{problem}
Theorem \ref{thm:seymour} can also be used to bound tree-depth of $P_t$-free graphs. The \emph{tree-depth} of a graph $G$ is the minimum height of a forest $F$ on the same vertex set with the property that for every edge of $G$, the corresponding vertices are in an ancestor-descendant relationship to each other in $F$ \cite{Nesetril12DeMendez}. 
\begin{corollary}
The tree-depth of a connected $P_t$-free graph $G$ of maximum degree $\Delta$ is at most $(t-2)(\Delta-1)+1$. 
\end{corollary}
Such a desired forest $F$ can be computed as follows. First compute an uncle tree $T$ for $G$. For each non-leaf vertex $v\in T$, create a path $P_v$ containing all the children of $v$. The forest $F$ is obtained by connecting the ``end''point of a path $P_v$ to the ``start''points of the paths of its children. Now recall that an uncle tree has height at most $t-1$ since each path from the root to a leaf is induced and that a non-root, non-leaf node has at most $\Delta-1$ children.  

\bibliographystyle{acm} 
\bibliography{colouringUpdated}




%
%
%
%
\end{document}